\documentclass[]{amsart}
\usepackage{qzart}
\usepackage{algorithmic}
\usepackage{framed}

\newenvironment{nscenter}
 {\parskip=0pt\par\nopagebreak\centering}
 {\par\noindent\ignorespacesafterend}
\numberwithin{equation}{section}

\begin{document}

\title[HANOVA]{Estimation and Prediction in Sparse and Unbalanced
  Tables}

\author{Qingyuan Zhao, Trevor Hastie and Daryl Pregibon}
\address{University of Pennsylvania, Stanford University and Google Inc.}

\email{qyzhao@wharton.upenn.edu}


\thanks{Qingyuan Zhao was a Ph.D.\ candidate at Stanford University
  when this work completed. Trevor Hastie was partially supported by
  grant DMS-1407548 from the National Science Foundation, grant 5R01
  EB 001988-21 from the National Institutes of Health, and a Google
  research grant.}

\begin{abstract}
We consider the problem where we have a multi-way table of means,
indexed by several factors, where each factor can have a large number
of levels. The entry in each cell is the mean of some response,
averaged over the observations falling into that cell. Some cells may
be very sparsely populated, and in extreme cases, not populated at
all. We might still like to estimate an expected response in such
cells. We propose here a novel hierarchical ANOVA (HANOVA) representation for
such data. Sparse cells will lean more on the lower-order interaction
model for the data. These in turn could have components that are
poorly represented in the data, in which case they rely on yet
lower-order models. Our approach leads to a simple hierarchical
algorithm, requiring repeated calculations of sub-table means of
modified counts. The algorithm has shown superiority over the unshrinked methods in both simulations and real data sets.
\end{abstract}

\maketitle

\section{Introduction}
\label{sec:introduction}

Prediction with factorial features has been studied for a long time
in statistics. Recently there are many arising
applications of this kind, but with much larger data size than
before. For instance, consider data on restaurant ratings, as provided
by Zagat, Yelp or other such services. In addition, the restaurants
can usually be classified according to various factors, such as
zip-code or geographical region (at least 10000 levels), type of cuisine
(Italian, French, etc, potentially dozens of levels), and price
category (often 1-5 stars). One important task is to estimate the
average rating for a particular kind of restaurants defined by these
factors, which can be used to answer questions like ``What is the most
popular cuisine in San Francisco?'' or ``Do expensive restaurants
usually get better ratings?''.

Such data sets and related questions have been studied thoroughly by
statisticians from the first formal proposal of analysis of
variance (ANOVA) in \citet{Fisher1918}. In most of the analysis of
variance literature, the focus is how to estimate the relative importance of
different variance components and how to test for their statistical
significance. The problem studied in this paper differs from classical
ANOVA in the following aspects:
\begin{itemize}
\item The focus is estimating and predicting cell means,
  whereas the main objective of ANOVA is to find important factors
  or interactions and test for their significance.
\item The observed data table is sparse and highly
  unbalanced. In many applications, we may only have a small proportion
  of the cells observed with quite different weights. For example, one
  zip code may have no Ethiopian
  restaurant, only one unpopular Japanese restaurant with a few
  dozen ratings, but many American restaurants with thousands of
  ratings in total.
\item The size of the data is usually very large and we also have
  a large number of main effects and interactions to estimate, so it
  is very important to find an computationally efficient algorithm
  while maintaining good statistical properties.
\end{itemize}

In the rest of the paper, we first illustrate our method with the
restaurant ratings example in Section \ref{sec:estimation-three-way},
then describe the statistical problem we are trying to solve and some
previous literature in Section \ref{sec:extens-gener-tabl}. In Section
\ref{sec:balanced} we prove some theoretical properties of our
algorithm when the observed data is balanced. This gives us an
empirical choice of penalty parameter, as described in Section
\ref{sec:unbalanced-table}. The rest of the paper is devoted to
algorithm implementation (Section \ref{sec:implementation}) and
simulations and some real data results (Section \ref{sec:results}).

\section{Estimation in Three-way Tables}
\label{sec:estimation-three-way}

Our proposed method is general can be applied to any number of
factors, but this would require a somewhat technical
representation. In this section we will restrict ourselves to three
factors ($F_1$, $F_2$ and $F_3$) to demonstrate the key elements in the algorithm.

Suppose these factors have $I$, $J$, and $K$ levels respectively. The
restaurant ratings are a collection of $n_{ijk}$  measurements  $y_{ijk\ell},\;\ell=1,\ldots,n_{ijk}$ at each cell $(i,j,k)$, and we are interested in predicting the cell mean $\mu_{i,j,k}=E(Y|F_1=i,F_2=j,F_3=k)$. We summarize our data by the cell means $\bar y_{ijk}=y_{ijk\cdot}/n_{ijk}$ and a weight $n_{ijk}$. For simplicity, we drop the bar, and refer to our aggregate data as $y_{ijk}$.

For some cells $n_{ijk}$ could be very small, even zero, but we would
still like to have a reasonable estimate.  In this case we would like
to shrink our estimate towards a more stable number that borrows
strength from similar cells.
This suggests a Bayesian mixed-effects framework \cite{Diggle1994},
where we assume the  $\mu_{ijk}$ are random, say Gaussian, with
distribution $\mu_{ijk}\sim N(\gamma_{ijk},\sigma_\mu^2)$. For the
moment assume $\gamma_{ijk}$ is known. If we assume the original
measurements $y_{ijk\ell}|\mu_{ijk}\sim N(\mu_{ijk},\sigma^2)$ and are
all independent, then the negative log-posterior likelihood for
$\mu_{ijk}$ given data is proportional to
\begin{equation}
  \label{eq:logpost}
  L(\mu)=\sum_{i,j,k}n_{ijk}(y_{ijk}-\mu_{ijk})^2 +\lambda \sum_{i,j,k}(\mu_{ijk}-\gamma_{ijk})^2,
\end{equation}
where $\lambda=\sigma^2/\sigma^2_\mu$. The posterior mode is simple to characterize:
\begin{equation}
  \label{eq:muhat}
 E(\mu_{ijk}|\mathbf{y})= \frac{n_{ijk}y_{ijk}}{n_{ijk}+\lambda} + \frac{\lambda \gamma_{ijk}}{n_{ijk}+\lambda},
\end{equation}
a simple weighted average of the observed mean $y_{ijk}$ and $\gamma_{ijk}$, with more emphasis on the former when $n_{ijk}$ is large.
Now $\gamma_{ijk}$ is not known, and so we can represent it by a parametric model $\gamma_{ijk}(\theta)$.
Empirical Bayes amounts to estimating $\theta$ by maximizing the marginal likelihood (integrating out $\mu$), and then using $\gamma_{ijk}(\hat\theta)$ as the target of shrinkage in (\ref{eq:muhat}).

\citet{Barry1990}, in a balanced two-way layout, proposed a simple main-effects parametrization $\gamma_{ij}=\alpha^0+\alpha^1_i+\alpha^2_j, $ with $\sum_i\alpha^1_i=0$, and $\sum_j\alpha^2_j=0$.
One can show in this case that the empirical Bayes procedure amounts to maximizing the following likelihood
\begin{equation}
  \label{eq:barry2}
L(\mu)=\sum_{i,j}(y_{ij}-\mu_{ij})^2+\lambda\sum_{i,j}(\mu_{ij}-\mu_{i\cdot}-\mu_{\cdot j}+\mu_{\cdot\cdot})^2,
\end{equation}
where $\mu_{i\cdot}=\sum_{i=1}^I\mu_{ij}/I$ etc.
Now one can show two things:
\begin{itemize}
\item We can write (\ref{eq:barry2}) in vector notation as
  \begin{equation}
    \label{eq:barry3}
    L(\bm{\mu})=(\mathbf{y}-\bm{\mu})^T(\mathbf{y}-\bm{\mu})+\lambda \bm{\mu}^T(\mathbf{I}-\mathbf{P}_A)\bm{\mu},
  \end{equation}
where $\mathbf{P}_A$ is the main-effects ANOVA projection operator. This has solution
\begin{equation}
  \label{eq:barry2sol}
  \hat{\bm{\mu}}=(\mathbf{I}+\lambda(\mathbf{I}-\mathbf{P}_A))^{-1}\mathbf{y}.
\end{equation}
\item This solution can be shown to have the much simpler form
\begin{equation}
  \label{eq:barry2sol}
  \hat{\bm{\mu}}=\frac{\mathbf{y}}{1+\lambda}+\frac{\lambda \tilde {\mathbf{y}}}{1+\lambda},
\end{equation}
where $\tilde{\mathbf{y}}=\mathbf{P}_A\mathbf{y}$, which is the main-effects ANOVA fit. This is of the same form as (\ref{eq:muhat}), with $n_{ij}=1$ and $\gamma_{ij}=\tilde y_{ij}$
\end{itemize}
This nice simplification disappears if we include weights in (\ref{eq:barry2}) or (\ref{eq:barry3}), although we do get a closed form expression for the solution along the lines of (\ref{eq:barry2sol}) (but involving weights).
What we like about (\ref{eq:barry2sol}), apart from its simplicity, is that it is easy to compute. The ANOVA fit $\tilde y_{ij}$ requires simple marginal means of $y_{ij}$ along the two factors, and the overall mean.
There are several problems though. This simplicity, both in
representation and especially in computation, goes away when each observation has weights. In addition, it may be that the main-effects ANOVA model is not well estimated for some values of $i$ or $j$, because of sparsity in these margins.

We propose a method that has the simplicity of (\ref{eq:muhat}) for a weighted model, that has a multi-level hierarchical structure, and is easy to compute.

\medskip

\begin{framed}
\begin{nscenter}
  {\bf \large Hierarchical penalized ANOVA model \\
 \center for a three-way table}
\end{nscenter}

\begin{enumerate}
\item[0.] Fit $\mu^{(0)}_{ijk}=\bar y_{...}$, the overall
  weighted mean of all the $y$s.
\item[1.] Fit the additive model $\mu^{(1)}_{ijk}=
  \alpha_i^1+\alpha_j^2+\alpha^3_k$ by solving the
  following weighted penalized least squares (WPLS) problem:
\begin{equation*}
\label{eq:addstep}
\min_{\alpha} \sum_{i,j,k}n_{ijk}(y_{ijk}-(\alpha_i^1+\alpha_j^2+\alpha^3_k))^2
+\lambda_{1} \cdot \sum_{n_{ijk} > 0}  (\alpha_i^1+\alpha_j^2+\alpha^3_k-\mu^{(0)}_{ijk})^2.
\end{equation*}
    This can be solved using a simple backfitting algorithm
    \cite{Hastie1990}, where each step is a version of
    (\ref{eq:muhat}) that only requires table summing. Here we shrink
    to the overall mean, so if a particular one-way marginal count is
    low, the shrinkage will still kick in.
  \item[2.] Fit the second-order interaction  model $\mu^{(2)}_{ijk}=
    \psi_{ij}^{12}+\psi_{jk}^{23}+\psi^{13}_{ik}$ by
    solving the WPLS problem:
    \begin{equation*}
\label{eq:intstep}
\begin{aligned}
\min_{\psi} &\sum_{i,j,k}n_{ijk}(y_{ijk}-(
\psi_{ij}^{12}+\psi_{jk}^{23}+\psi^{13}_{ik}))^2 \\
&+\lambda_{2} \cdot \sum_{n_{ijk} > 0} (
\psi_{ij}^{12}+\psi_{jk}^{23}+\psi^{13}_{ik}-\mu^{(1)}_{ijk})^2.
\end{aligned}
    \end{equation*}
    This is the case we were concerned
    about in Section \ref{sec:introduction}, when some two-way tables
    have sparse counts. In this case, that term is shrunk more towards its parent.
  \item [3.] Fit the full third-order model $\mu^{(3)}_{ijk}$ by solving the WPLS problem
     \begin{equation*}
\label{eq:intstep}
\min_{\mu} \sum_{i,j,k}n_{ijk}(y_{ijk}-\mu_{ijk})^2
+\lambda_{3}\sum_{n_{ijk}>0} ( \mu_{ijk}-\mu^{(2)}_{ijk})^2.
    \end{equation*}
    This is where we started, i.e. model~(\ref{eq:muhat}), with $\gamma_{ijk}=\tilde y^{(2)}_{ijk}$. The term we shrink to has been regularized to accommodate lower-level sparsity.
\end{enumerate}
\end{framed}

We have omitted the details of backfitting in steps 1 and 2. Consider
using coordinate-descent to solve (\ref{eq:addstep}), where for
example we hold $\alpha^2_j$ and $\alpha^3_k$ fixed at their current
estimates and solve for $\alpha^1_i$. By collecting terms, this is again of the form (\ref{eq:muhat}), and so the update involves simple table averaging of a modified response and shrinkage target.

There are many interesting details that haven't been discussed
yet. For example, we have the same $\lambda$ at each level; these
could be different. We can take the empirical Bayes analogy further
and use variance components to suggest values for $\lambda$, or at
least {\em relative values}. We will discuss this in Section
\ref{sec:theory}. Another option is to terminate the hierarchy
early. In the restaurant ratings example, three way table is
manageable even with many levels for each factor. However, when the
number of factors grows, the number of cells grows at a exponential
rate. In this case we might truncate the hierarchy early,  e.g. at
step 2, and use $\tilde y^{(2)}_{ijk}=\tilde \psi_{ij}^{12}+\tilde\psi_{jk}^{23}+\tilde\psi^{13}_{ik}$ as the estimate. Note in this case we would only need to store and compute at the level of two-way tables.

\section{Extension to General Tables}
\label{sec:extens-gener-tabl}

Let's extend our algorithm to a general table with $m \ge 1$
factors. All the cell means are represented by a vector $\mathbf
y$. Each cell mean is indexed by its factor levels
$I=(i_1,i_2,\cdots,i_m)$. For example, if $m=3$, then the cell that is
in the 2nd, 4th and 1st level of the three factors corresponds to
$y_I = y_{241}$. Every cell is also associated with a weight, denoted by $n_{I}$.

Let $\Omega$ be the set of observed indices. The general statistical
question we are trying to address is: suppose the ratings follow an
underlying statistical model $\mathbf y = \boldsymbol{\mu} +
\boldsymbol{\epsilon}$ and $\mu_I = f(I; \boldsymbol{\beta})$ for some
deterministic or random function $f$, if we have only observed
$\mathbf{y}_\Omega$, how do we estimate the underlying cell means
$\boldsymbol{\mu}$?

\subsection{General Algorithm}
\label{sec:general-algorithm}
In this paper, our algorithm assumes a hierarchy of models $f_0, f_1,
\cdots f_m$ increasing in complexity. Each time the algorithm uses
fits from the previous model $f_k$ as the prior for a new and
more complex model $f_{k+1}$. This is a very general hierarchical
model and we may choose different sequence of model functions in
different applications.

In our problem, the model function $f_k,~0 \le k \le m$ is a
$k$-th order interaction model, i.e.
\begin{equation}
  \label{eq:modelfunction}
  \mu^{(k)}_I = f_k(I;\boldsymbol{\beta}^{(k)}) = \sum_{J \subset [m],~|J|=k} \beta^J_{I_J}
\end{equation}
Here $J$ is a subset of $[m] = \{1,2,\ldots,m\}$, indicating the
factors we are considering; $\beta^J_{I_J}$ is the "mean'' effect of
cells that are in levels $I_J$ of the factors specified in $J$, and
$\boldsymbol{\beta}^{(k)} = \{ \boldsymbol{\beta}^J:~J \subset
[m],|J|=k\}$ is the model parameter for $f_k$.

Our algorithm is summarized in Algorithm \ref{alg:hbmodel}. Notice
that this idea can be generalized to solve any hierarchical model, not
necessarily the full additive model in \eqref{eq:modelfunction}.

\begin{algorithm}[H]
\caption{Generic algorithm to solve \eqref{eq:modelfunction}}
\label{alg:hbmodel}
\begin{algorithmic}
\FOR {$k = 1,\ldots,m+1$}
\STATE Estimate $\boldsymbol{\beta}^{(k)}$ by
$\hat{\boldsymbol{\beta}}^{(k)}$, the Bayes estimate given
$\mathbf{y}$ (data) and $\boldsymbol{\beta}^{(k-1)} =
\hat{\boldsymbol{\beta}}^{(k-1)}$ (prior).
\ENDFOR
\end{algorithmic}
\end{algorithm}

\subsection{Previous Methods}
\label{sec:previous-methods}

Many statisticians studied this kind of problem in the 1980's using
Bayes models. One typical model \citep{DuMouchelHarris1983} is
\begin{equation}
\begin{aligned}
y_{ij} &= \theta_{ij} + \epsilon_{ij},\\
\theta_{ij} &= \mu + \alpha_i + \beta_{j} + \delta_{ij},
\end{aligned}
\end{equation}
and they also want to extrapolate for some missing cells in the table. The empirical Bayes method is to put a normal prior on $\alpha$ and $\beta$
\begin{equation}
\alpha,\beta \sim \mathrm{N}(0,V)
\end{equation}
and estimate the hyperprior variance parameter $V$ using the data. A more general Bayesian ANOVA model and many more examples are
considered in the tutorial paper by \citet{Casella1992}. See also \citet[ch. 9]{Searle1992} and \citet{Searle2006}.

More recently, \citet{Beran2005} considered ANOVA from a pure
shrinkage and optimization view. The author considers the whole class of
estimators that solve penalized least squares, which may correspond to
some Bayes modeling or not. Then among all the possible
shrinkage estimators, the paper tries to find the one that minimize estimated
risk. This is a very sound approach from a theoretical point of view, but
the algorithm proposed in the paper scales poorly and couldn't be used
in real data sets.

In \cite{Gelman2005}, the author uses a hierarchical Bayes model to analyze Analysis of Variance in an universal way. The statistical model is
\begin{equation}
y_i = \sum_{m=0}^M \beta^{(m)}_{j_i^m},
\end{equation}
where $m$ stands for a batch of regression coefficients (or in terms
of ANOVA, variance component). Then he assumes normal priors on these coefficients
\begin{equation}
\beta_j^{(m)} \sim \mathrm{N}(0,\sigma_m^2), \quad \forall j=1,\ldots,J_m,~m=1,\ldots,M
\end{equation}
and puts hierarchical prior on $\sigma_m^2$. This bears some
similarity to our statistical model described in
\ref{sec:general-algorithm}. Gelman's paper focuses mostly on the
relationship between this hierarchical Bayes model with classical
ANOVA and how to use Gibbs sampler to sample from the posterior
distribution of $\sigma_m^2$.

Following this idea, a more recent paper by \citet{Volfovsky2012} considers a three-way factorial model
\begin{equation}
y_{ijkl} = \mu + a_i + b_j + c_k + (ab)_{ij} + (ac)_{ik} + (bc)_{jk} + (abc)_{ijk} + \epsilon_{ijkl}
\end{equation}
where $\epsilon_{ijkl}$ are i.i.d. $\mathrm{N}(0,\sigma^2)$. However,
the effects here are possibly correlated, which makes the problem different
from the usual settings. The authors put a normal prior with zero mean
and general covariance matrix on the effects and treat the covariance matrices as parameters with invert-Wishart priors. By plugging in empirical Bayes estimate for the hyperpriors, they can run a Gibbs sampler to find the posterior.

However, despite the flexibility of hierarchical Bayes model, it
usually needs a Markov Chain Monte Carlo (MCMC) algorithm to sample
from posterior distribution. This is unacceptable for more than a few
hundred cells. Also, as noticed in Section \ref{sec:introduction}, our
focus is on estimating and predicting cell means instead of
inference about variance components.

From the perspective of high dimensional estimation, shrinkage is known
to be very effective, for example the James-Stein
estimator dominates maximum likelihood estimator
\citep{Efron2010}. Recently, in a multi-task averaging problem,
\citet{Feldman2012} uses James-Stein estimator with empirically
estimated covariance matrix and claims to outperform James-Stein in
some occasions. However, the multi-task averaging problem has no
associated covariates, and the structural information of these factors
are crucial in our algorithm.

In light of the hierarchical Bayes model considered so far, our
algorithm can be viewed as an empirical solution to a hierarchical Bayes model.

\section{Theory}
\label{sec:theory}

In this section we develop some theoretical results for HANOVA. These
also provide a method to empirically select the regularization
parameter $\lambda$ in the algorithm.

\subsection{Balanced Table}
\label{sec:balanced}

Let's first assume all the observed cells have the same weight
(i.e. same number of observations).\footnote{This is different from
what a "balanced table' is meant in most of the statistics literature. In
our paper, we say a table is balanced if all the observed cells have
the same number of observations, but the table could contain many
empty cells.}
Assume we observe $n$ cells and the responses are centered. At each stage, we are maximizing the log "likelihood'' function
\begin{equation} \label{eq:loglike}
-\frac{1}{2} [\|\mathbf{y} - \boldsymbol{\mu}^{(k)}\|^2 + \lambda_k \|\boldsymbol{\mu}^{(k)} - \boldsymbol{\mu}^{(k-1)}\|^2 ]
\end{equation}
subject to $\boldsymbol{\mu}^{(k)} \in \mathcal{S}_k$. Here $\mathcal{S}_k$ is the linear subspace of $\mathbb{R}^n$ generated by all possible $k$-th order effects. This implies the subspaces are nested $\mathcal{S}_1 \subset \mathcal{S}_2 \subset \ldots \subset \mathcal{S}_k$. Say the dimension of $\mathcal{S}_k$ is $d_k$ and the projection matrix onto $\mathcal{S}_k$ is $\mathbf{P}_k$. This means $\mathbf{P}_k = \mathbf{U}_k \mathbf{U}_k^T$ where the columns of $\mathbf{U}_k \in \mathbb{R}^{n \times d_k}$ are orthonormal basis for $\mathcal{S}_k$, i.e. $\mathbf{U}_k^T \mathbf{U}_k = \mathbf{I}_{d_k}$. Because the subspaces are nested, we can assume the first $d_{k-1}$ columns of $\mathbf{U}_k$ are $\mathbf{U}_{k-1}$, i.e. $\mathbf{U}_k = \left( \begin{array} {cc} \mathbf{U}_{k-1} & \mathbf{V}_k \end{array} \right)$. Moreover suppose $\mathbf{U} = \left( \begin{array}{cc} \mathbf{U}_K & \mathbf{V} \end{array} \right) \in \mathbb{R}^{n \times n}$ is an orthogonal matrix.

Consider a hierarchical Bayes model that is a special case of \eqref{eq:modelfunction}.
\begin{equation} \label{eq:hiermodel}
\begin{aligned}
\mathbf{y} | \boldsymbol{\mu} &\sim \mathrm{N}(\boldsymbol{\mu},\sigma^2 \mathbf{I}_n),~\sigma^2 \mathrm{~is~known}\\
\boldsymbol{\mu} | \boldsymbol{\beta}^{(m)} &\sim
\mathrm{N}(\mathbf{U}_m \boldsymbol{\beta}^{(m)},\sigma_m^2 \mathbf{I}_n), \\
\boldsymbol{\beta}^{(m)} | \boldsymbol{\beta}^{(m-1)} &\sim \mathrm{N}(\mathbf{U}_{m}^T\boldsymbol{\mu}^{(m-1)},\sigma_{m-1}^2 \mathbf{I}_{d_m}),~ \boldsymbol{\mu}^{(m)} = \mathbf{U}_m \boldsymbol{\beta}^{(m)}\\
&~~\vdots  \\
\boldsymbol{\beta}^{(k)} | \boldsymbol{\beta}^{(k-1)} &\sim \mathrm{N}(\mathbf{U}_{k}^T\boldsymbol{\mu}^{(k-1)},\sigma_{k-1}^2 \mathbf{I}_{d_k}),~\boldsymbol{\mu}^{(k)} = \mathbf{U}_k \boldsymbol{\beta}^{(k)} \\
&~~\vdots  \\
\boldsymbol{\beta}^{(1)} | \boldsymbol{\beta}^{(0)} &\sim \mathrm{N}(\mathbf{U}_{1}^T\boldsymbol{\mu}^{(0)},\sigma_{0}^2 \mathbf{I}_{d_1}),~\boldsymbol{\mu}^{(1)} = \mathbf{U}_1 \boldsymbol{\beta}^{(1)}. \\
\end{aligned}
\end{equation}

As the next theorem indicates, this hierarchical model is closely
related to maximizing the log-likelihood \eqref{eq:loglike}. In the
usual ANOVA model, each factor may have totally different main effects
(or interactions) and we can test for the significance of
them. However, as pointed out in Section \ref{sec:introduction}, we
are no longer interested in any individual factor. By using orthogonal
matrices in the hierarchical model \eqref{eq:hiermodel}, we
implicitly treat the effects and interactions of the observable factors as
"randomly chosen" directions in the column space of $\mathbf{U}_m$.

\begin{theorem} \label{thm:empBayes1}
  For balanced table, each step of the hierarchical procedure
  (Algorithm \ref{alg:hbmodel}) for
  \eqref{eq:hiermodel} is equivalent to maximizing
  \eqref{eq:loglike} subject to $\boldsymbol{\mu}^{(k)} \in
  \mathcal{S}_k$, with
  \begin{equation} \label{eq:lambdak}
  \lambda_k = \frac{\sigma^2 + \sigma_m^2 + \ldots +
    \sigma_k^2}{\sigma_{k-1}^2}
  \end{equation}
  and $\boldsymbol{\beta}^{(k-1)}$ being the fit from last step.
\end{theorem}
\begin{proof}
  See Appendix \ref{app:empBayes1}.
\end{proof}

Theorem \ref{thm:empBayes1} also indicates we can estimate all the
variance parameters $\sigma_k^2$ from the data and then compute a
vector of $\bm{\lambda}$. Let's assume the ratings are centered (by the global mean)
so we put $\boldsymbol{\beta}^{(0)} = \mathbf{0}$. Then the multi-level
hierarchical model \eqref{eq:hiermodel} is actually a linear random effects model
\begin{equation}
\label{eq:randeff}
\begin{aligned}
\mathbf{y} | \boldsymbol{\mu} &\sim \mathrm{N}(\boldsymbol{\mu},\sigma^2 \mathbf{I}_n),~\sigma^2 \mathrm{is~known}\\
\boldsymbol{\mu} | \boldsymbol{\beta}^{(m)} &\sim \mathrm{N}(\mathbf{U}_m\boldsymbol{\beta}^{(m)},\sigma_m^2 \mathbf{I}_n), \\
\boldsymbol{\beta}^{(m)} &\sim \mathrm{N}(\mathbf{0},\boldsymbol{\Sigma}),\\
\boldsymbol{\Sigma} &= \left(
\begin{array}{cccc}
\tau_0^2 \mathbf{I}_{d_1} & & & \\
& \tau_1^2 \mathbf{I}_{d_2 - d_1} & & \\
& & \ddots & \\
& & & \tau_{m-1}^2 \mathbf{I}_{d_m - d_{m-1}} \\
\end{array}
\right), \\
\tau_i^2 &= (\sigma_i^2 + \ldots + \sigma_{m-1}^2), ~i=0,\cdots,m.
\end{aligned}
\end{equation}

Now we can state our main theorem that guarantees the effectiveness of
our algorithm for balanced table.
\begin{theorem} \label{thm:empBayes2}
For balanced table, if we use \eqref{eq:lambdak} to compute $\bm{\lambda}$,
the solution to our algorithm is the posterior mean of $\boldsymbol{\beta}$ in the linear random effects model \eqref{eq:randeff}.
\end{theorem}
\begin{proof}
  See Appendix \ref{app:empBayes2}.
\end{proof}

In reality we don't know what $\sigma_0^2$, $\sigma_1^2$,\ldots,
$\sigma_m^2$ are, but we can use various methods developed in ANOVA to
estimate them. By doing this, our algorithm is equivalent to an
empirical Bayes solution of the multi-level hierarchical model.

Following the general principle of variance component analysis, we can
compute the expectation of certain quadratic forms of $\mathbf{y}$, then use these quadratic forms to find an unbiased estimator of $\sigma_k^2$. For example,
\begin{equation} \label{eq:estimatingequation1}
\begin{aligned}
\mathrm{E}[\mathbf{y}^T \mathbf{U}_k \mathbf{U}_k^T \mathbf{y}]
&= \mathrm{tr}[\mathbf{U}_k \mathbf{U}_k^T \mathrm{Var}(\mathbf{y})]\\
&= \mathrm{tr}[\mathbf{U}_k \mathbf{U}_k^T (\sigma^2 \mathbf{I}_n + \sigma_m^2 \mathbf{I}_n + \mathbf{U}_m\boldsymbol{\Sigma}\mathbf{U}_m^T)]\\
&= d_k (\sigma^2 + \sigma_m^2) +  \sum_{j=0}^{k-1}(d_{j+1} - d_j)
\tau_j^2, \\
& \forall~k = 1, \cdots,m
\end{aligned}
\end{equation}
Also
\begin{equation} \label{eq:estimatingequation2}
  \begin{aligned}
\mathrm{E}[\mathbf{y}^T \mathbf{U} \mathbf{U}^T \mathbf{y}] &= \mathrm{E}
[\|\mathbf{y}\|^2] \\
& = \mathrm{tr}[(\sigma^2+\sigma_m^2) \mathbf{I}_n +
\mathbf{U}_m \boldsymbol{\Sigma} \mathbf{U}_m^T] \\
& = n (\sigma^2 + \sigma_m^2) + \sum_{j=0}^{m-1} d_{j+1} \sigma_j^2
  \end{aligned}
\end{equation}
Thus we have $m+1$ equations and $m+1$ parameters (recall we assume
$\sigma^2$ is known), the unbiased estimator of the prior variances
can be obtained by solving this linear system.

In fact, if we call $\mathbf{U}_{m+1} = \mathbf{U}$ and $d_{m+1} = n$,
we have
\begin{equation}
  \begin{aligned}
\mathrm{E}[\mathbf{y}^T \mathbf{U}_{k+1} \mathbf{U}_{k+1}^T \mathbf{y} -
\mathbf{y}^T \mathbf{U}_k \mathbf{U}_k^T \mathbf{y}]
= (d_{k+1} - d_k) \tau_k^2, \\
\forall~k=0,\cdots, m
  \end{aligned}
\end{equation}
So
\begin{equation} \label{eq:sigmahat}
\begin{aligned}
\hat{\sigma}^2_k =& \frac{\mathbf{y}^T \mathbf{U}_{k+1} \mathbf{U}_{k+1}^T \mathbf{y} -
\mathbf{y}^T \mathbf{U}_k \mathbf{U}_k^T \mathbf{y}}{d_{k+1} - d_k} \\
&- \frac{\mathbf{y}^T \mathbf{U}_{k+2} \mathbf{U}_{k+2}^T \mathbf{y} -
\mathbf{y}^T \mathbf{U}_{k+1} \mathbf{U}_{k+1}^T \mathbf{y}}{d_{k+2} -
d_{k+1}}, \\
& \qquad \qquad \qquad \qquad k=0,\cdots,m-1 \\
\hat{\sigma}^2_m =& \frac{\mathbf{y}^T \mathbf{y} - \mathbf{y}^T \mathbf{U}_m
\mathbf{U}_m^T \mathbf{y}}{n - d_m} - \sigma^2
\end{aligned}
\end{equation}
are unbiased for estimating $\sigma_k^2$. Note that it is also
possible to treat $\sigma^2$ as a tuning parameter, but usually $\sigma^2$ can be estimated quite accurately from the data.

The above formulae is only one specific (perhaps the easiest) choice of
estimating equations. More general methods in analysis of variance for
unbalanced tables, such as Restricted Maximum Likelihood (REML) can
also be used (see, for example, \citet{Searle1992}).

\subsection{Unbalanced Table}
\label{sec:unbalanced-table}


If the observed table is unbalanced, then we cannot prove any exact
conclusions like \ref{thm:empBayes1} or \ref{thm:empBayes2}. In this
case, the log likelihood function we maximize is
\begin{equation} \label{eq:loglikeunbalanced}
-\frac{1}{2} [(\mathbf{y} - \boldsymbol{\mu}^{(k)})^T \mathbf{N} (\mathbf{y} - \boldsymbol{\mu}^{(k)}) + \lambda_k \|\boldsymbol{\mu}^{(k)} - \boldsymbol{\mu}^{(k-1)}\|^2 ]
\end{equation}
where $\mathbf{N}$ is a diagonal matrix with weights.

Similar to the previous section, we may sequentially maximize
\eqref{eq:loglikeunbalanced} for $k=1,2,\ldots,m$ to estimate $\bm{\mu}$ and $\boldsymbol{\beta}$ in the following linear random effects model
\begin{equation}
\label{eq:randeffunbalanced}
\begin{aligned}
\mathbf{y} | \boldsymbol{\mu} &\sim \mathrm{N}(\boldsymbol{\mu},\sigma^2 \mathbf{N}^{-1}),~\sigma^2 \mathrm{is~known}\\
\boldsymbol{\mu} | \boldsymbol{\beta}^{(m)} &\sim \mathrm{N}(\mathbf{U}_m\boldsymbol{\beta}^{(m)},\sigma_m^2 \mathbf{I}_n), \\
\boldsymbol{\beta}^{(m)} &\sim \mathrm{N}(\mathbf{0},\boldsymbol{\Sigma}),\\
\end{aligned}
\end{equation}
Here $\boldsymbol{\Sigma}$ is the same as the one
\eqref{eq:randeff}. This model is simply the unbalanced
version of \eqref{eq:randeff}.

The theory we derived in Section \ref{sec:balanced} can be used to
give us a vector of reasonable penalty parameters $\bm{\lambda}$. One way to do this is to pretend all the observed cells have the
same weight and use equations in \eqref{eq:sigmahat} to estimate all the
$\sigma_k^2$. Then one can compute ratios between the variance
components (i.e. \eqref{eq:lambdak}) to get $\bm{\lambda}$. The estimate
$\hat{\sigma}_k^2$ is still unbiased, because the estimation equations
\eqref{eq:estimatingequation1} and \eqref{eq:estimatingequation2} still hold true after replacing $\sigma^2$
with $\sigma^2 \mathrm{tr}(\mathbf{N}^{-1}/n)$.

\section{Implementation}
\label{sec:implementation}

In this section we discuss some implementation details of HANOVA.

\subsection{Preprocessing}
\label{sec:preprocessing}

In real applications, many data sets are not summarized in the format
of cell means and weights. Usually we may have multiple units in a
cell and every unit receive multiple reviews. A mixed-effect model
describing this is
\begin{equation}
\label{eq:singlecell}
\begin{aligned}
  y_{ci} &\sim \mathrm{N}(\mu_C + \alpha_i, \sigma_{r}^2/n_i), \\
  \alpha_i &\sim \mathrm{N}(0, \sigma_u^2).
\end{aligned}
\end{equation}
Here $\mu_c$ is the cell mean, $\alpha_i$ is the unit effect,
$\sigma_{r}^2$ is the variance of user's rating. $\sigma_r$ may actually depend
on the unit, and the procedure below can be slightly modified to this
heterogeneous case. $n_i$ is the number of
reviewers of that restaurant, $\sigma_u^2$ is the variance of the
unit effect.

The preprocessing procedure for this model is described in Algorithm
\ref{alg:preprocess}. Note that since we use the estimated variance
directly as weight in step 3, the $\sigma^2$ defined in
\eqref{eq:hiermodel} is $1$.

\begin{algorithm}[hbtp]
\caption{Data Preprocessing for HANOVA}
\label{alg:preprocess}
\begin{algorithmic}[1]
\STATE Estimate $\sigma_u^2$ and $\sigma_r^2$ based on all the observed
  units.
\STATE Plug $\hat{\sigma}_u^2$ and $\hat{\sigma}_r^2$ in
  \eqref{eq:singlecell} and estimate the $\mu_c$ by
\[
\hat{\mu}_c = \frac{\sum_{i} \frac{y_{ci}}{\hat{\sigma}_u^2 +
    \hat{\sigma}_r^2/n_i}}{\sum_{i} \frac{1}{\hat{\sigma}_u^2 +
    \hat{\sigma}_r^2/n_i}}
\sim \mathrm{N}(\mu_c, \frac{1}{\sum_{i} \frac{1}{\hat{\sigma}_u^2 +
    \hat{\sigma}_r^2/n_i}})
\]
\STATE Use $\hat{\mu}_C$ and $\sum_{i} \frac{1}{\hat{\sigma}_u^2 +
    \hat{\sigma}_r^2/n_i}$ as value and weight in HANOVA to obtain
  regularized estimate of $\mu_c$, denote by $\hat{\hat{\mu}}_c$.
\STATE Estimate individual restaurant effect $\alpha_i$ by shrinking
  $y_{ci}$ towards $\hat{\hat{\mu}}_c$, using the following formula
\[
\hat{\hat{\mu}} + \hat{\alpha_{i}} = \frac{n_iy_{ci}/\hat{\sigma}_r^2 +
  \hat{\hat{\mu}}_{c}/\hat{\sigma}_u^2}{n_{i}/\hat{\sigma}_r^2 + 1/\hat{\sigma}_u^2}
\]
\end{algorithmic}
\end{algorithm}

\subsection{Normal Equations}
In our algorithm, the $k$-th order model is
\begin{equation*} \label{objective}
\min_{\boldsymbol{\beta}^{(k)}} \sum_{I \in \Omega} n_I (y_I -
\mu^{(k)}_I)^2 + \lambda_k \sum_{I \in \Omega} (\mu_I^{(k)} - \mu_I^{(k-1)})^2
\end{equation*}
Here $\mu^{(k)}_I = \sum_{J \subset [m],|J|=k}
\beta^{J}_{I_J}$. Differentiate the above loss with respect to each
$\mu^{J}_L$ and equate to zero
\begin{equation*}
  \begin{aligned}
\sum_{I_J = L} (n_I + \lambda_k) \sum_{K \subset [m],|K|=k}
\beta^{K}_{I_K} =& \sum_{I_J = L} (n_I y_I + \lambda_k
\mu^{(k-1)}_I),\\
&\quad \forall~J \subset [m],|J|=k,L.
  \end{aligned}
\end{equation*}
So the coefficient of $\beta^{K}_{M}$ in the $(J,L)$-th equation is
\begin{equation*}
z^{J,L}_{K,M} =
\left\{
\begin{aligned}
&0&,&K=J,~ M \ne L\\
&\sum_{I_J = L} (n_I + \lambda_k)&,& K=J,~M=L\\
&\sum_{I_J = L, I_K=M} (n_I + \lambda_k)&,& K \ne J
\end{aligned}
\right.
\end{equation*}

Now the problem reduces to solve a large system of linear equations
\begin{equation*}
\sum_{K,M} z^{J,L}_{K,M} \beta^{K}_M = \sum_{I_J = L} (n_I y_I + \lambda_k \mu^{(k-1)}_I),\quad \forall J \subset [m],|J|=k,L.
\end{equation*}

\subsection{Backfitting}
We implemented the block coordinate descent algorithm (or backfitting
algorithm ) to solve the above equations. Each block contains
$\beta^J_\cdot$, i.e. all the effects for certain margins. For
example, when $J$ contains only one factor, this block contains all
the main effects associated with that factor. The psuedocode of our
algorithm is in Algorithm \ref{alg:pseudocode}.

Note that the basic building block of every iteration is
computing $\sum_{I_J = L} (n_I + \lambda_k) \mu_I(\mathrm{old})$. For a
fixed $J$, this amounts to a weighted table sum over $J^C$. This
operation can be easily parallelized.

\begin{algorithm}[hbtp]
\caption{psuedo-code of HANOVA}
\label{alg:pseudocode}
\begin{algorithmic}[1]
\STATE Choose $\mathrm{maxk}$ from $1,\cdots,m$.

\FOR {$k= 1 \to \mathrm{maxk}$}

  \FORALL {$|J| = k,~L \in \mathcal{L}(J)$}
  \STATE initialize $\beta^J_L \leftarrow \frac{\bar{y}}{{m \choose k}}$
  \STATE $u^J_L \leftarrow \sum_{I_J = L} (n_I y_I + \lambda_k \mu^{(k-1)}_I)$
  \STATE $z^{J,L}_{J,L} \leftarrow \sum_{I_J = L} (n_I + \lambda_k)$
  \ENDFOR

  \STATE $\boldsymbol{\mu}^{(k)} = \boldsymbol{\mu}^{(k-1)}$

  \REPEAT
  \FORALL {$|J|=k$}
      \STATE $\mathbf{s}^{J} \leftarrow \mathrm{apply}(\boldsymbol{\mu}^{(k)}(\mathbf{n}+\lambda_k),J^C,\mathrm{sum})$
      \STATE $\boldsymbol{\beta}^{J}_{\cdot} \leftarrow \boldsymbol{\beta}^{J}_{\cdot} + (\mathbf{u}^{J}_{\cdot} - \mathbf{s}^{J}_{\cdot}) / \mathbf{z}^{J,\cdot}_{J,\cdot}$
      \STATE Update $\boldsymbol{\mu}^{(k)}$
  \ENDFOR
\UNTIL{$\boldsymbol{\mu}$ converges.}
\ENDFOR
\end{algorithmic}
\end{algorithm}

\section{Results}
\label{sec:results}
\subsection{Simulations}

We first use simulations to verify the optimality of HANOVA. We
simulate $200$ instances from the linear mixed effect model
\eqref{eq:randeff} with four factors (each has 10 levels) and the true model contains 2-way interactions. The variance
parameters are chosen to be $(\sigma_0,\sigma_1,\sigma_2,\sigma) =
(2,1,0,0.5),\mathrm{or~}(2,1,0,1),\mathrm{or~}(1,2,0,1)$ (fairly
strong signal). The empirically estimated
$\lambda_1$ sometimes can be infinity. In this case, we
truncate $\lambda_1$ to $5$.

The results are shown in the violin plots in Figure \ref{fig:sub1} to \ref{fig:sub3}. In
all three cases HANOVA with oracle $\bm{\lambda}$ achieves Bayes risk and
HANOVA with empirical $\bm{\lambda}$ performs almost as good as the
oracle. The unshrinked linear model suffers from increased noise,
as we can see from the first and second plots.

The next simulation suggests HANOVA can fit a high order model
without overfitting. In this simulation, the true model
is of order $3$ but the three way interactions are very weak
($(\sigma_0,\sigma_1,\sigma_2,\sigma_3,\sigma) =
(2,1,\mathbf{0.5},0,1)$). The cells can have up to $10$
times different weights.

We generate $50$ instances and plot the estimation RMSE in Figure
\ref{fig:sub4}. The unshrinked linear model with all the
three-way interactions performs poorly, due to overfitting. The second
order HANOVA model is slightly better than second order linear model,
and the third order model is able to squeeze a little more.

\begin{figure}[htbp]
  \centering
  \begin{subfigure}[b]{0.48\textwidth}
    \includegraphics[width=\textwidth]{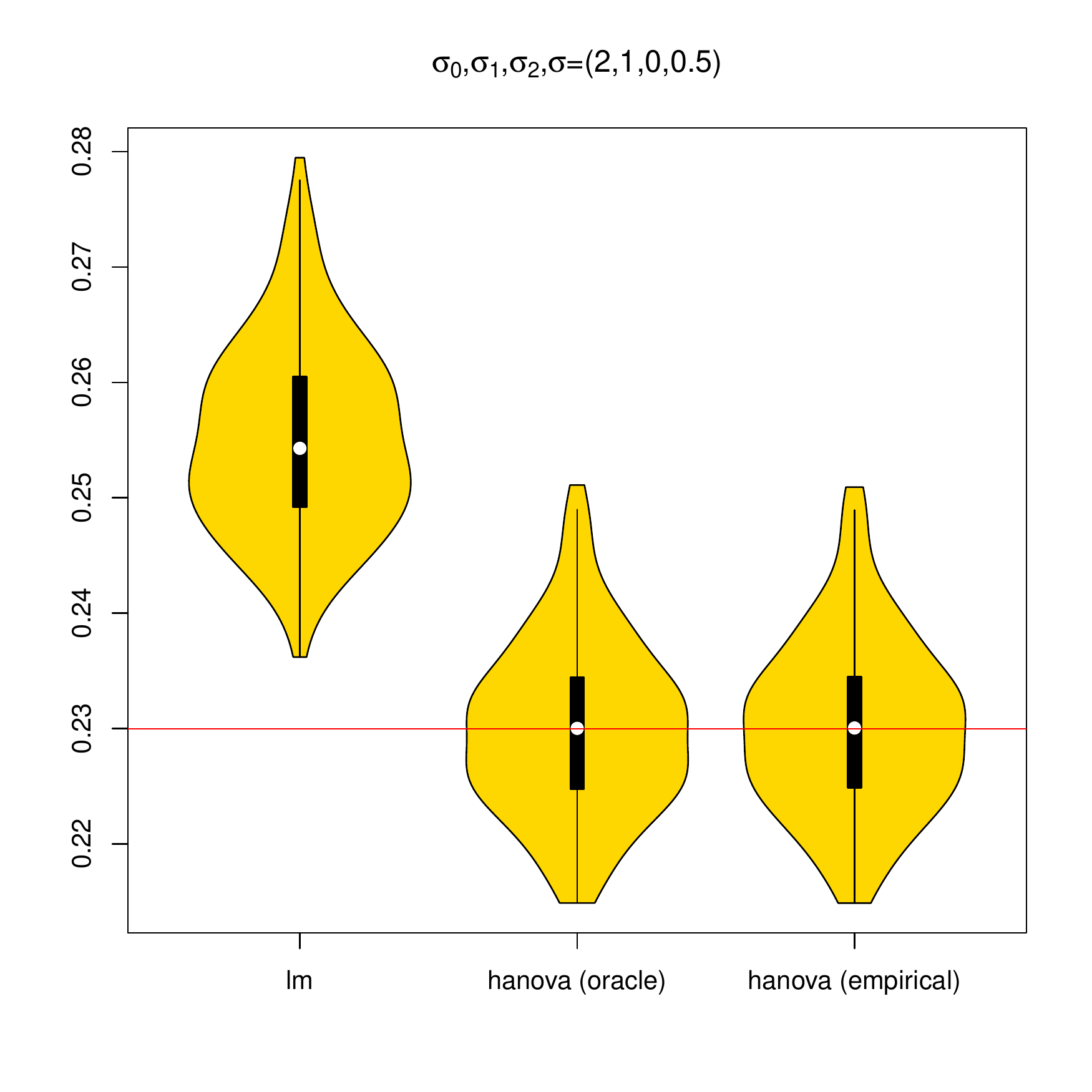}
    \caption{Balanced table with small noise.}
    \label{fig:sub1}
  \end{subfigure}
  ~
  \begin{subfigure}[b]{0.48\textwidth}
    \includegraphics[width=\textwidth]{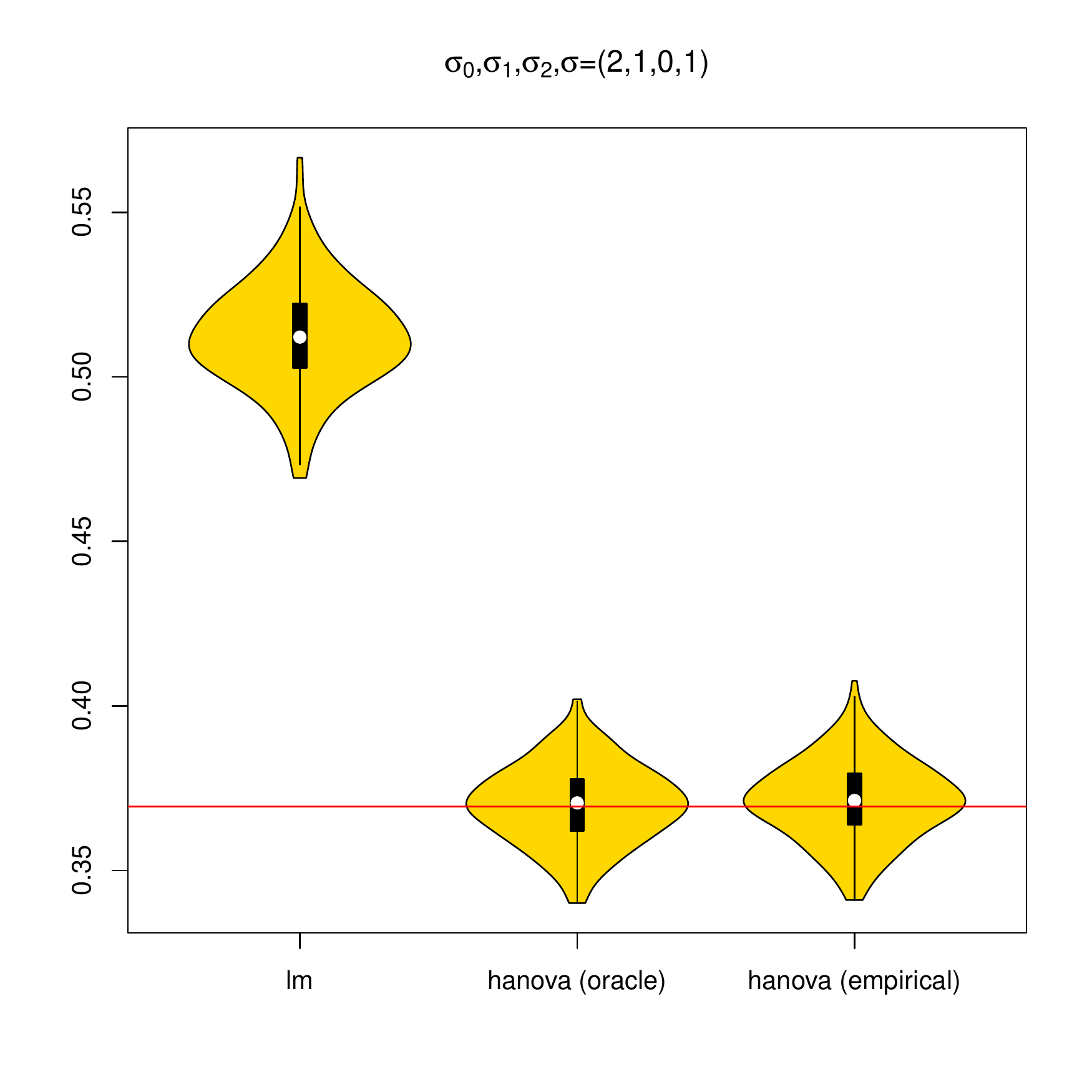}
    \caption{Balanced table with medium noise.}
    \label{fig:sub2}
  \end{subfigure}

  \begin{subfigure}[b]{0.48\textwidth}
    \includegraphics[width=\textwidth]{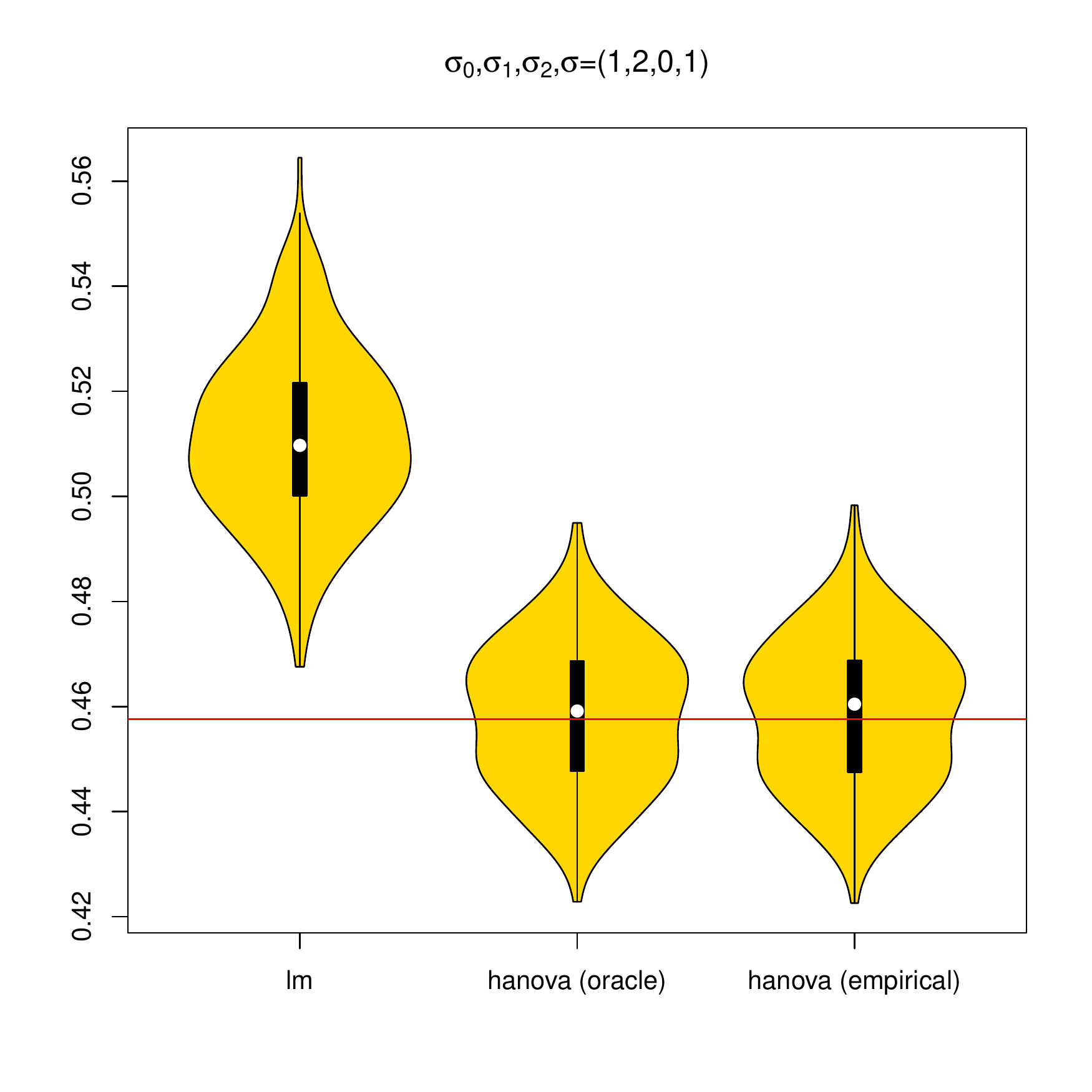}
    \caption{Balanced table with strong interactions.}
    \label{fig:sub3}
  \end{subfigure}
  ~
  \begin{subfigure}[b]{0.48\textwidth}
    \includegraphics[width=\textwidth]{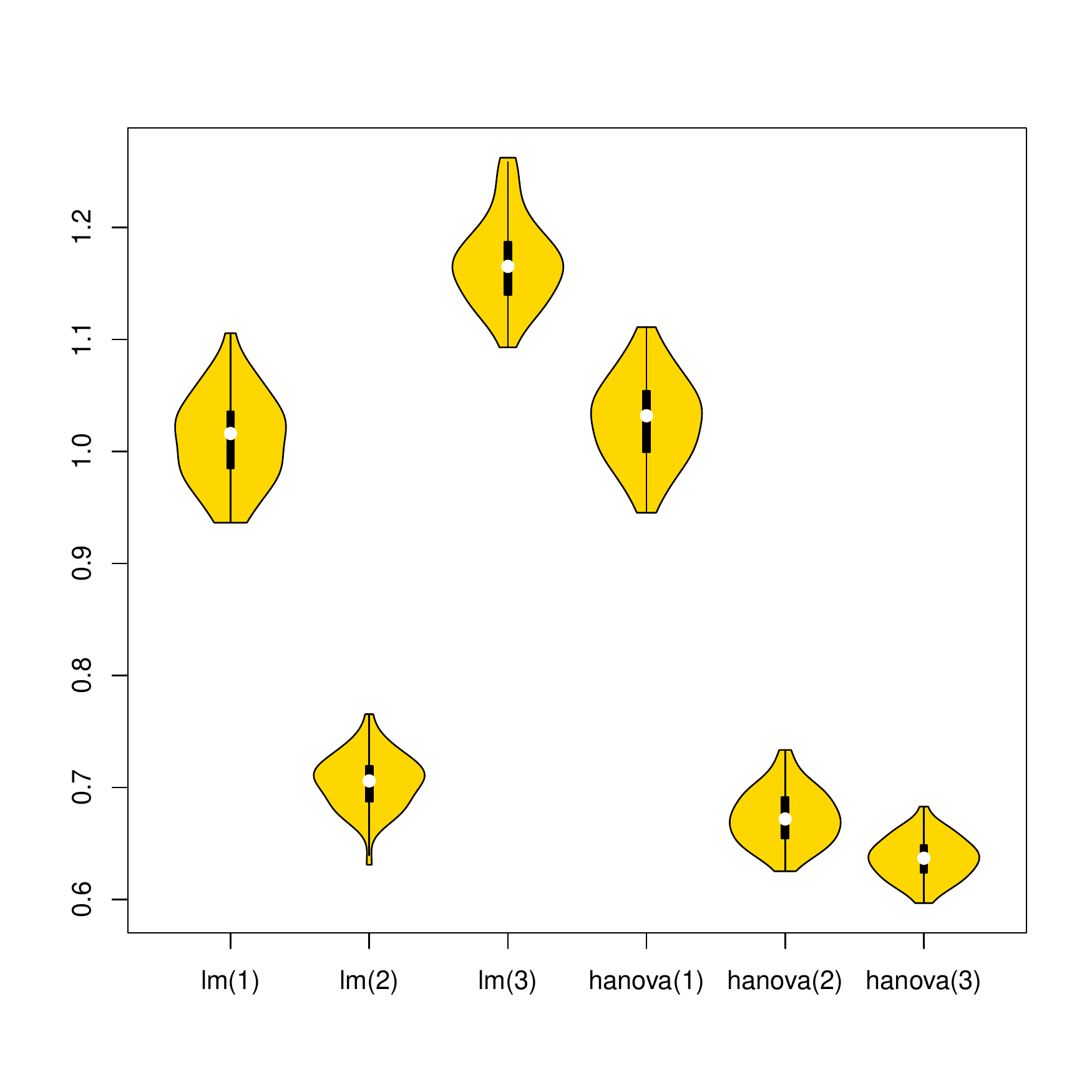}
    \caption{Unbalanced table with third-order interactions.}
    \label{fig:sub4}
  \end{subfigure}

  \caption{Comparison of HANOVA with linear model in various
    simulation settings. The three methods being compared in the first
    three plots are (left to
    right) linear model with two-way interactions, HANOVA with oracle
  $\bm{\lambda}$, and HANOVA with empirically estimated
  $\bm{\lambda}$. $y$-axis is the Root Mean Squared Error on all the
  observed cells. The red horizontal line is the Bayes risk. The
  methods being compared in the bottom-right plot are linear model (order 1, 2, 3), and HANOVA (order 1, 2, 3)
    with empirically estimated $\bm{\lambda}$.}
  \label{fig:1}
\end{figure}

\subsection{Real Data}

\subsubsection{IMDb Data}
\label{sec:imdb}

We ran our algorithm on an IMDb movie ratings data set. This data set
contains all the directors who have more than 20 movies on IMDb
record. The data set is a $5 \times 20 \times 33$ table with 276
observations (about $8\%$ of all the cells). The three factors
in the data set are decade, genre and director.

Since the observed cell is too sparse,  we will only fit our
hierarchical penalized model to first order and compare it with the
unshrinked main effects model. The procedure described in Section
\ref{sec:unbalanced-table} gives a suggested $\lambda_1 = 0.064$. We
further ran cross-validation on the training set and choose to use
$\lambda_1 = 0.1$. We use the squared root of the number of movies as
cell weights.

The first order HANOVA with $\lambda_1$ has test set Root Mean Squared
Error (RMSE)$=0.798$, while the unshrinked main effects model has
$\mathrm{RMSE} = 0.804$ and the cross-validated ridge regression main
effects model has $\mathrm{RMSE} = 0.801$.

\subsubsection{Zagat Data}
\label{sec:zagat-data}

This data set contains survey ratings of 1776 restaurants in New York
City in $33$ zip-codes, $25$ different type of cuisines and $7$ price
levels. Each rating is from $0$ to $30$. $20\%$ of the restaurants are left out as test data.

The standard deviation of reviewer random effect is about $7.51$. This large
random effect makes the ratings of many unpopular restaurants not trustable, as pointed out
in Section \ref{sec:introduction}. The number of reviews of each
restaurants follow a heavy-tailed distribution. One restaurant has
$8922$ reviews but $106$ restaurants have less than or equal to $5$
reviews. The average and median number of reviews are $117.1$ and
$42.0$.

The standard deviation of restaurant random effect (within cell) is
about $2.47$. This is the lower limit of RMSE of any prediction
methods based on the three factors. We use the
method described in \ref{sec:preprocessing} to preprocess the data.

The $\sigma_m$ estimated by HANOVA (defined in \eqref{eq:hiermodel}) is
$(0.423, 0.149, 0)$ (recall the $\sigma^2$ is always $1$ after the
preprocessing procedure \ref{alg:preprocess}) and the corresponding $\lambda = (5.71, 45.0,
\inf)$. This is the case that the data has fairly strong main effects
and weak interactions. Since the third-order interactions are too
weak, we only fit the first three HANOVA models. The RMSEs are $3.010$
(HANOVA grand-mean model), $2.641$ (HANOVA main effects model) and
$2.638$ (HANOVA second-order interactions model). The best $\lambda_2$
selected by cross-validation is $72.5$ and the corresponding
second-order interactions model has RMSE$=2.625$. As comparison, the
RMSE of random main effects model is $2.640$ and the RMSE of price-only random effects model is $2.767$.

The RMSE reduction of HANOVA interactions model is more significant
in cells that have fewer restaurants. This is one of HANOVA's the main
purposes.
Among the $355$ restaurants in the
test data set, $32$ of them receive rating adjustment greater than 1 from
main effects model to interactions model. There mean absolute
prediction error for these $32$ restaurants are reduced from $2.031$ to
$1.707$ by fitting a interactions model by HANOVA.

\bibliographystyle{plainnat}
\bibliography{ref}

\appendix

\section{Proof of Theorems}
\subsection{Theorem \ref{thm:empBayes1}}
\label{app:empBayes1}

First look at
\begin{equation}
\boldsymbol{\beta}^{(k)} | \boldsymbol{\beta}^{(k-1)} \sim \mathrm{N}(\mathbf{U}_{k}^T\boldsymbol{\mu}^{(k-1)},\sigma_{k-1}^2 \mathbf{I}_{d_k}),~\boldsymbol{\mu}^{(k-1)} = \mathbf{U}_{k-1} \boldsymbol{\beta}^{(k-1)}
\end{equation}
The log of the density function is
\begin{equation}
\begin{aligned}
\log f(\boldsymbol{\beta}^{(k)} | \boldsymbol{\beta}^{(k-1)}) &=
- \frac{1}{2} \frac{\| \boldsymbol{\beta}^{(k)} - \mathbf{U}_k^T
\boldsymbol{\mu}^{(k-1)} \|^2}{\sigma_{k-1}^2} + \mathrm{const} \\
&= - \frac{1}{2} \frac{\| \mathbf{U}_k (\boldsymbol{\beta}^{(k)} - \mathbf{U}_k^T
\boldsymbol{\mu}^{(k-1)}) \|^2 }{\sigma_{k-1}^2}  + \mathrm{const} \\
&= -\frac{1}{2} \frac{\| \boldsymbol{\mu}^{(k)} -
\boldsymbol{\mu}^{(k-1)} \|^2}{\sigma_{k-1}^2} + \mathrm{const}
\end{aligned}
\end{equation}

The distribution of $\mathbf{y}$ given $\boldsymbol{\beta}_{k}$ is
\begin{equation}
\mathbf{y} | \boldsymbol{\beta}_{k} \sim \mathrm{N}(\mathbf{U}_{k} \boldsymbol{\beta}_{k}, \sigma^2 \mathbf{I}_n + \sigma_m^2 \mathbf{I}_n + \mathbf{U}_m \mathbf{D} \mathbf{U}_m^T).
\end{equation}
Here $\mathbf{D}$ is a diagonal matrix
\begin{equation}
\mathbf{D} = \left(
\begin{array}{ccc}
(\sigma_{m-1}^2 + \ldots + \sigma_k^2) \mathbf{I}_{d_k} & & \\
& \ddots & \\
& & \sigma_{m-1}^2 \mathbf{I}_{d_m - d_{m-1}} \\
\end{array}
\right).
\end{equation}
Thus
\begin{equation}
\begin{aligned}
\log f(\mathbf{y} | \boldsymbol{\beta}_{k}) &= -\frac{1}{2}
(\mathbf{y} - \mathbf{U}_k \boldsymbol{\beta}^{(k)})^T
((\sigma^2+\sigma_m^2) \mathbf{I}_n + \mathbf{U}_m \mathbf{D}
\mathbf{U}_m^T)^{-1} (\mathbf{y} - \mathbf{U}_k
\boldsymbol{\beta}^{(k)}) \\
& \qquad \qquad \qquad \qquad \qquad \qquad \qquad \qquad
\qquad
+ \mathrm{const}\\
&= -\frac{1}{2} (\mathbf{y} - \mathbf{U}_k \boldsymbol{\beta}^{(k)})^T \left[\mathbf{U} \left(
\begin{array}{cc}
(\sigma^2 + \sigma_m^2 + \ldots \sigma_k^2) \mathbf{I}_{d_k} & \\
& \tilde{\mathbf{D}} \\
\end{array}
\right)
\mathbf{U}^T\right]^{-1} \\
& \qquad \qquad \qquad \qquad \qquad \qquad \qquad \qquad
\qquad
\cdot (\mathbf{y} - \mathbf{U}_k
\boldsymbol{\beta}^{(k)}) + \mathrm{const} \\
&= -\frac{1}{2} (\mathbf{y} - \mathbf{U}_k \boldsymbol{\beta}^{(k)})^T \mathbf{U} \left(
\begin{array}{cc}
(\sigma^2 + \sigma_m^2 + \ldots \sigma_k^2) \mathbf{I}_{d_k} & \\
& \tilde{\mathbf{D}} \\
\end{array}
\right)^{-1} \\
& \qquad \qquad \qquad \qquad \qquad \qquad \qquad \qquad
\cdot \mathbf{U}^T(\mathbf{y} - \mathbf{U}_k \boldsymbol{\beta}^{(k)}) + \mathrm{const} \\
\end{aligned}
\end{equation}
Notice that $\mathbf{U}^T (\mathbf{y} - \mathbf{U}_k \boldsymbol{\beta}^{(k)}) = \left(
\begin{array}{c}
\mathbf{U}_k^T \mathbf{y} - \boldsymbol{\beta}^{(k)} \\
\mathbf{U}_k^{\perp^T} \mathbf{y} \\
 \end{array}
\right)$, so
\begin{equation}
\log f(\mathbf{y} | \boldsymbol{\beta}^{(k)}) = -\frac{\|\boldsymbol{\beta}^{(k)} - \mathbf{U}_k^T \mathbf{y}\|^2}{2(\sigma^2 + \sigma_K^2 + \ldots + \sigma_k^2)} - g(\mathbf{y}).
\end{equation}
The claim is immediately followed by the fact that $\|\boldsymbol{\beta}^{(k)} - \mathbf{U}_k^T \mathbf{y}\|^2 = \|\mathbf{U}_k(\boldsymbol{\beta}^{(k)} - \mathbf{U}_k^T \mathbf{y})\|^2 = \|\boldsymbol{\mu}^{(k)} - \mathbf{P}_k \mathbf{y}\|^2 = \| \boldsymbol{\mu}^{(k)} - \mathbf{y} \|^2 - \|\mathbf{y} - \mathbf{P}_k\mathbf{y}\|^2$. The regularization parameter is
\begin{equation}\lambda_k = \frac{\sigma^2 + \sigma_K^2 + \ldots + \sigma_k^2}{\sigma_{k-1}^2}\end{equation}.

\subsection{Theorem \ref{thm:empBayes2}}
\label{app:empBayes2}
Recall our algorithm is just iteratively maximizing likelihood \eqref{eq:loglike}
to get $\hat{\boldsymbol{\beta}}_m$ and then find the posterior mean
of $\boldsymbol{\beta}$ given $\mathbf{y}$
and $\boldsymbol{\beta}^{(k)} = \hat{\boldsymbol{\beta}}_m$.

By empirical Bayes estimator of $\boldsymbol{\beta}$, I mean just
compute the posterior mean of $\boldsymbol{\beta}$ given $\mathbf{y}$
in \eqref{eq:randeff} and plug in whatever $\sigma^2$ and $\sigma_k^2$ we used to compute $\lambda_k$.

The posterior mean of \eqref{eq:randeff} can be computed by Tweedie's
formula\cite{Robbins1964}. The marginal distribution of $\mathbf{y}$ is
$\mathrm{N}(\mathbf{0}, (\sigma^2+\sigma_m^2) \mathbf{I}_n +
\mathbf{U}_m \boldsymbol{\Sigma} \mathbf{U}_m^T)$, i.e.
\begin{equation}
m(\mathbf{y}) \propto \exp \{ - \frac{1}{2} \mathbf{y}^T ((\sigma^2+\sigma_m^2) \mathbf{I}_n +
\mathbf{U}_m \boldsymbol{\Sigma} \mathbf{U}_m^T)^{-1} \mathbf{y} \}
\end{equation}
so the posterior mean is given by
\begin{equation}
\begin{aligned} \label{eq:bayessol}
\mathrm{E}[\boldsymbol{\mu} | \mathbf{y}]
&= \mathbf{y} + \sigma^2 \nabla \log m(\mathbf{y}) \\
&= \mathbf{y} - \sigma^2 ((\sigma^2+\sigma_m^2) \mathbf{I}_n +
\mathbf{U}_m \boldsymbol{\Sigma} \mathbf{U}_m^T)^{-1} \mathbf{y} \\
&=  \frac{\sigma_m^2}{\sigma^2 + \sigma_m^2} \mathbf{y} \\
& \quad +
\mathbf{U} \left(\frac{\sigma^2}{\sigma^2 + \sigma_m^2}\mathbf{I}_n - \sigma^2
\left(
\begin{array}{cc}
(\sigma^2 + \sigma_m^2) \mathbf{I}_{d_m} + \boldsymbol{\Sigma} & \\
& (\sigma^2 + \sigma_m^2) \mathbf{I}_{n-d_m} \\
\end{array}
\right)^{-1} \right) \mathbf{U}^T  \mathbf{y} \\
&= \frac{\sigma_m^2}{\sigma^2 + \sigma_m^2} \mathbf{y} +
\sum_{k=1}^m \left(\frac{\sigma^2}{\sigma^2 + \sigma_m^2} - \frac{\sigma^2}{\sigma^2 + \sigma_m^2 +
    \cdots + \sigma_{k-1}^2} \right) \mathbf{V}_k \mathbf{V}_k^T
\mathbf{y}
\end{aligned}
\end{equation}

On the other hand, in the case of a balanced table, each step of our
algorithm produces
\begin{equation}
\begin{aligned}
  \hat{\boldsymbol{\mu}}^{(k)} &= \mathbf{U}_k \mathbf{U}_k^T
(\frac{1}{1+\lambda_k} \mathbf{y} + \frac{\lambda_k}{1 + \lambda_k}
\hat{\boldsymbol{\mu}}^{(k-1)}) \\
&= \mathbf{U}_k \mathbf{U}_k^T \left( \frac{\sigma^2_{k-1}}{\sigma^2 +
  \sigma_m^2 + \cdots + \sigma_{k-1}^2} \mathbf{y} + \frac{\sigma^2 +
  \sigma_m^2 + \cdots + \sigma_{k}^2}{\sigma^2 +
  \sigma_m^2 + \cdots + \sigma_{k-1}^2} \hat{\boldsymbol{\mu}}^{(k-1)}
\right), \\
& \qquad\qquad\qquad\qquad\qquad\qquad\qquad\qquad\qquad\qquad  \quad k=1,\ldots,m
\end{aligned}
\end{equation}
and finally
\begin{equation} \label{eq:hiersol}
\hat{\boldsymbol{\mu}} = \mathrm{E} [\boldsymbol{\mu} |
\mathbf{y}, \hat{\mathbf{y}}^{(k)}] = \frac{\sigma_m^2}{\sigma^2 +
  \sigma_m^2} \mathbf{y} + \frac{\sigma^2}{\sigma^2 + \sigma_m^2} \hat{\mathbf{y}}^{(m)}
\end{equation}

It suffices to show \eqref{eq:bayessol} and \eqref{eq:hiersol} are
actually the same. Notice that $\mathbf{U}_k \mathbf{U}_k^T
\hat{\boldsymbol{\mu}}^{(j)} = \hat{\boldsymbol{\mu}}^{(j)}$ for $j < k$ because
$\hat{\boldsymbol{\mu}}^{(j)} \in \mathcal{S}_j \subset \mathcal{S}_k$ and
$\mathbf{U}_k \mathbf{U}_k^T$ is just the projection matrix onto
$\mathcal{S}_k$. With some calculation,
\begin{equation}
\begin{aligned}
\hat{\boldsymbol{\mu}}^{(k)} &= \sum_{k=1}^m \frac{\sigma^2 +
  \sigma_m^2}{\sigma^2 + \sigma_m^2 + \cdots + \sigma_k^2}
\frac{\sigma_{k-1}^2}{\sigma^2 + \sigma_m^2 + \cdots + \sigma_{k-1}^2}
\mathbf{U}_{k} \mathbf{U}_{k}^T \mathbf{y}\\
&= (\sigma^2 +   \sigma_m^2)\sum_{k=1}^m \left(\frac{1}{\sigma^2 + \sigma_m^2 + \cdots + \sigma_k^2} - \frac{1}{\sigma^2 + \sigma_m^2 + \cdots + \sigma_{k-1}^2}\right)
\mathbf{U}_k \mathbf{U}_k^T \mathbf{y}
\end{aligned}
\end{equation}
Since $\mathbf{U}_k \mathbf{U}_k^T \mathbf{y} = \sum_{j=1}^k
\mathbf{V}_j \mathbf{V}_j^T \mathbf{y}$, it is easy to see that
\begin{equation}
\frac{\sigma^2}{\sigma^2 + \sigma_m^2} \hat{\boldsymbol{\mu}}^{(k)}
= \sigma^2 \sum_{k=1}^m \left(\frac{1}{\sigma^2 + \sigma_m^2} - \frac{1}{\sigma^2 + \sigma_m^2 + \cdots + \sigma_{k-1}^2}\right)
\mathbf{V}_k \mathbf{V}_k^T \mathbf{y}
\end{equation}
The theorem is immediately proved if we plug this in \eqref{eq:hiersol}
and compare it to \eqref{eq:bayessol}.

\end{document}